\documentclass[11pt]{amsart}

\usepackage[centertags]{amsmath}
\usepackage{hyperref}
\usepackage{epsfig}
\usepackage{amssymb,latexsym}
\usepackage{array}
\newcolumntype{L}[1]{>{\raggedright\let\newline\\\arraybackslash\hspace{0pt}}m{#1}}
\newcolumntype{C}[1]{>{\centering\let\newline\\\arraybackslash\hspace{0pt}}m{#1}}
\usepackage{indentfirst}
\usepackage{float,subfig}
\usepackage{caption}
\usepackage{mathrsfs}
\usepackage{amsthm}
\usepackage{amsmath}
\usepackage{amssymb}
\usepackage{amsfonts}
\usepackage{amsbsy}
\usepackage{appendix}
\usepackage{bbding}
\usepackage{tikz}
\usetikzlibrary{matrix,fit,shapes.misc,positioning}
\tikzset{%
  highlight/.style={rectangle,rounded corners,fill=red!15,draw,fill opacity=0.3,thick,inner sep=0pt}
}
\tikzset{%
  highlight1/.style={rectangle,rounded corners,fill=blue!15,draw,fill opacity=0.3,thick,inner sep=0pt}
}

\theoremstyle{plain}
\newtheorem{thm}{Theorem}[section]

\newtheorem{lem}[thm]{Lemma}
\newtheorem{prop}[thm]{Proposition}

\theoremstyle{definition}

\newtheorem{rem}[thm]{Remark}

\numberwithin{equation}{section}
\newcommand{\Fq}{\mathbb{F}_{q}}

\newcommand{\F}{{\mathbb F}}
\newcommand{\Tr}{{\rm Tr}}

\newcommand{\dsum}{\displaystyle \sum}
\newcommand{\lcm}{\mbox{lcm}}
\begin{document}

\title[Quasi-Cyclic Subcodes of Cyclic Codes]{Quasi-Cyclic Subcodes of Cyclic Codes}

%\begin{center}
\author{Jean-Claude Belfiore}
\address{ Dept. Comelec, T\'{e}l\'{e}com ParisTech, 46 rue Barrault, 75634 Paris Cedex 13, France}
\email{belfiore@enst.fr}

\author{Cem G\"uner\.{i}\textsuperscript{1}}
\address{Sabanc{\i} University, FENS, 34956  \.{I}stanbul, Turkey}
\email{guneri@sabanciuniv.edu}

\author{Buket \"Ozkaya}
\address{ Dept. Comelec, T\'{e}l\'{e}com ParisTech, 46 rue Barrault, 75634 Paris Cedex 13, France}
\email{buket.ozkaya@telecom-paristech.fr}
%\symbolfootnotetext{Corresponding author.}
\vspace{0.4cm}

\footnotetext[1]{Cem G\"{u}neri is supported by T\"{U}B\.{I}TAK
project 114F432.} 

\begin{abstract}
We completely characterize possible indices of quasi-cyclic subcodes
in a cyclic code for a very broad class of cyclic codes. We present
enumeration results for quasi-cyclic subcodes of a fixed index and
show that the problem of enumeration is equivalent to enumeration of
certain vector subspaces in finite fields. In particular, we present
enumeration results for quasi-cyclic subcodes of the simplex code
and duals of certain BCH codes. Our results are based on the
trace representation of cyclic codes.
\end{abstract}
\maketitle

\section{Introduction}
Let $\Fq$ denote the finite field with $q$ elements, where $q$ is a
prime power, and let $m$ and $\ell$ be positive integers. A linear
code $C \subseteq \Fq^{m\ell}$ is called a quasi-cyclic (QC) code of
index $\ell$ if it is invariant under shift of codewords by $\ell$
units and $\ell$ is the minimal number with this property. Clearly
QC codes are generalizations of cyclic codes, for which $\ell=1$. QC
codes drew much attention in the literature since they yield codes
with good parameters (see for instance \cite{C,DH}). The class of QC
codes and some of its subclasses also perform well asymptotically
and reach the Gilbert-Varshamov bound (\cite{D,K,LS1,MW}).

Studying subcodes in well-known classes of codes is a common theme
in coding theory for various purposes. Our motivation to study QC
subcodes in cyclic codes stems from \cite{GM}, where the number of
rational points of supersingular curves is related to weight
analysis of certain subcodes of cyclic codes. It is shown in
\cite{GM} that these subcodes are QC codes.

We consider cyclic codes of length $q^n-1$ over $\F_q$ and assume
throughout that the dual code's zeros all have $q$-cyclotomic cosets
of length $n$ over the base field $\F_q$ (cf. Section
\ref{indices}). Note that this is true for a broad class of cyclic
codes. We have two particular problems addressed: to determine all
possible indices of QC subcodes in a given cyclic code and to count
the number of QC subcodes for a fixed index. We solve the first
problem completely and list all positive integers that are indices
of some QC subcode (Theorem \ref{index theorem}). In particular, we
observe that not every divisor of the cyclic code's length need to
be the index of some QC subcode (Remark \ref{example}). For the
second problem, we show that the enumeration of QC subcodes in a
cyclic code is related to the count of vector subspaces in finite
fields. For the class of cyclic codes we study, we show that these
two problems are equivalent (cf. Theorem \ref{enum-main}). Using
this observation, we can count QC subcodes of a given index in
certain well-known cyclic codes. Enumeration results require
counting subspaces in $\F_{q^n}$ which are defined over a subfield
$\F_{q^d}$ in a maximal way (i.e. these spaces do not have a vector
space structure over a subfield that contains $\F_{q^d}$). We give
exact count of such vector spaces in Theorem \ref{max field}, using
inclusion-exclusion principle. We utilize the trace representation
of cyclic codes (\cite{W}) to obtain our results.

Organization of the paper is as follows: We determine possible
indices of QC subcodes in a given cyclic code in Section
\ref{indices}. Enumeration of QC subcodes is addressed in Section
\ref{enum}, where the relation to counting vector subspaces in a
finite field is given. In Section \ref{simplex}, we enumerate QC
subcodes of the $q$-ary simplex code of length $q^n-1$ for any prime
power $q$ and any $n$. Section \ref{bch} contains enumeration
results for QC subcodes of duals certain BCH codes. Proofs of our
enumeration results yield an algorithm that counts indices and their
appearances (multiplicities) for certain cyclic codes. Section
\ref{examples} consists of some examples produced by the algorithm. Magma code of the algorithm is made available on-line for interested readers (\cite{Program}).

\section{Indices of QC Subcodes}\label{indices}

Let $n$ and $N$ be positive integers with $N= q^n-1$ and let
$\alpha$ a primitive $N^{th}$ root of unity. Throughout this work,
we will concentrate on the cyclic code $C$ over $\Fq$ of length $N$
with basic dual zeros
$$BZ(C^{\bot})=\{\alpha^{i_1},\ldots,\alpha^{i_s}\},$$
where $i_j\geq 1$ for all $j$ and $i_j$'s come from pairwise
distinct $q$-cyclotomic cosets mod $N$. This means that the
generating polynomial of $C^\bot$ is the product of the minimal
polynomials of $\alpha^{i_j}$'s over $\F_q$. Since $N$ and $q$ are
relatively prime these minimal polynomials are distinct.
\textbf{Moreover, we will assume throughout that the $q$-cyclotomic
coset mod $N$ for each $i_j$ has size $n$.} Note that this amounts
to saying that the minimal polynomials of $\alpha^{i_j}$'s over
$\F_q$ are all of degree $n$; or equivalently
$\F_{q^n}=\F_q(\alpha^{i_j})$ for each $j$.

Trace representation of $C$ is as follows (\cite[Proposition
2.1]{W}):
\begin{equation}\label{cyclic trace}
C=\left\{
\left(\Tr_{\F_{q^n}/\F_q}\left(\lambda_1\alpha^{ki_{1}}+\cdots+
\lambda_s\alpha^{ki_{s}}\right) \right)_{0\leq k \leq N-1} ;
\lambda_j \in \F_{q^n},\ 1 \leq j \leq s\right\}. \end{equation}
Coordinates of length $N$ codewords of $C$ are obtained by
evaluating the trace expression for each $0\leq k \leq N-1$.

Consider a subcode of $C$,
\begin{equation}\label{subcode}
C'=\left\{\left(\Tr_{\F_{q^n}/\F_q}\left(\beta_1\alpha^{ki_{1}}+\cdots+
\beta_s\alpha^{ki_{s}}\right) \right)_{0\leq k \leq N-1} ; \beta_j
\in V_j\subseteq \F_{q^n},\ 1 \leq j \leq s\right\}.
\end{equation}

The following result will play a crucial role in this article.

\begin{thm}\cite[Theorem 2.5]{G} \label{cg}
Let $i_j\geq 1$ be positive integers (for $1\leq j \leq s$) which
are in different $q$-cyclotomic cosets mod $N=q^n-1$. For
$\lambda_1,\ldots ,\lambda_s\in \F_{q^n}$, we have
\begin{equation}\label{old}
\Tr_{\F_{q^n}/ \F_q}\left(\lambda_1 x^{i_1} + \cdots + \lambda_s
x^{i_s}\right)=0, \ \ \mbox{for all $x\in \F_{q^n}$}
\end{equation}
if and only if each $i_j$ has $q$-cyclotomic coset mod $N$ of length
$|Cyc_q(i_j)|=\delta_j <n$ and $\Tr_{\F_{q^n}/
\F_{q^{\delta_j}}}(\lambda_j)=0$ for all $j=1,\ldots ,s$. In
particular, if each $i_j$ has $q$-cyclotomic coset of length $n$,
then (\ref{old}) holds if and only if $\lambda_j=0$ for all $j$.
\end{thm}

Theorem \ref{cg} justifies our assumption on the sizes of
$q$-cyclotomic cosets of $i_j$'s in the trace representation of the
cyclic code $C$, as can be seen in the next results and in Section
\ref{enum} when we consider enumeration of QC subcodes.

\begin{lem}
$C'$ is an $\Fq$-linear subcode if and only if $V_j$ is an
$\Fq$-linear subspace of $\F_{q^n}$ for all $j$.
\end{lem}

\begin{proof}
Choose two arbitrary codewords from $C'$:
$$v_{\beta}=\left(\Tr_{\F_{q^n}/\F_q}\left(\dsum_{j=1}^s \beta_j\alpha^{ki_{j}}\right) \right)_{0\leq k \leq N-1}, \quad v_{\gamma}=\left(\Tr_{\F_{q^n}/\F_q}\left(\dsum_{j=1}^s
\gamma_j\alpha^{ki_{j}}\right) \right)_{0\leq k \leq N-1} \in C',$$
(i.e. $\beta_j,\gamma_j \in V_j$ for all $j$). $C'$ is an
$\Fq$-linear subcode of $C$ if and only if for any $a\in \F_q$ we
have
$$av_{\beta}+v_{\gamma}= \left(\Tr_{\F_{q^n}/\F_q}\left(\dsum_{j=1}^s
(a\beta_j+\gamma_j)\alpha^{ki_{j}}\right) \right)_{0\leq k \leq N-1}
\in C'.$$ That is, there exist $\rho_j\in V_j$ (for all $j$) such
that
$$\left(\Tr_{\F_{q^n}/\F_q}\left(\dsum_{j=1}^s
(a\beta_j+\gamma_j)\alpha^{ki_{j}}\right) \right)_{0\leq k \leq
N-1}=\left(\Tr_{\F_{q^n}/\F_q}\left(\dsum_{j=1}^s \rho_j
\alpha^{ki_{j}}\right) \right)_{0\leq k \leq N-1},$$ or equivalently
\begin{equation}\label{condition}
\Tr_{\F_{q^n}/\F_q}\left(\dsum_{j=1}^s
\left[(a\beta_j+\gamma_j)-\rho_j\right] \alpha^{ki_{j}}\right)=0 \ \
\mbox{for all $0\leq k \leq N-1$}.
\end{equation}
Since we assumed that every $i_j$ has $q$-cyclotomic coset mod $N$
of size $n$, Theorem \ref{cg} implies that (\ref{condition}) holds
if and only if  $a\beta_j+\gamma_j=\rho_j$ for all $j$ (i.e.
$a\beta_j+\gamma_j \in V_j$ for all $j$). The result follows.
\end{proof}

\textbf{We will assume from now on that $V_j$'s describing the
subcode $C'$ in (\ref{subcode}) are all $\F_q$-subspaces of
$\F_{q^n}$.} Let $T$ denote the cyclic shift operator on $\F_q^N$,
i.e. $T(u_1,u_2,\ldots ,u_N)=(u_N,u_1,\ldots ,u_{N-1})$. The
following is easy to observe.
\begin{lem}\label{shift}
$$T\left[\left(\Tr_{\F_{q^n}/\F_q}\left(\dsum_{j=1}^s
\beta_j\alpha^{ki_{j}}\right) \right)_{0\leq k \leq
N-1}\right]=\left(\Tr_{\F_{q^n}/\F_q}\left(\dsum_{j=1}^s
(\beta_j\alpha^{-i_{j}})\alpha^{ki_{j}}\right) \right)_{0\leq k \leq
N-1}.$$
\end{lem}

For each $1 \leq j \leq s$, let us now define the following subcode
of $C'$:
$$C'_j= \left\{\left(\Tr_{\F_{q^n}/\F_q}\left(\beta_j\alpha^{ki_{j}}\right)\right)_{0\leq k \leq N-1} ; \beta_j
\in V_j\right\}.$$

Next, we obtain a criterion for quasi-cyclicity of $C_j'$.

\begin{prop} \label{ff0}
$C'_j$ is an index $\ell_j$ QC code if and only if $V_j$ is an
$\Fq(\alpha^{\ell_j i_{j}})$-subspace of $\F_{q^n}$ and $\ell_j$ is
the minimal such number.
\end{prop}

\begin{proof}
According to Lemma \ref{shift}, $C'_j$ is closed under shift of
codewords by $t$ units if $V_j$ is closed under multiplication by
$\alpha^{-ti_{j}}$. Since $V_j$ is an $\Fq$-subspace of $\F_{q^n}$,
this is equivalent to saying that $V_j$ is closed under scalar
multiplication by elements in
$$\Fq[\alpha^{-ti_{j}}]=\Fq(\alpha^{-ti_{j}})=\Fq(\alpha^{ti_{j}}).$$
Hence, if $V_j$ is an $\Fq(\alpha^{\ell_j i_{j}})$-subspace of
$\F_{q^n}$ and $\ell_j$ is the minimal such number, then $C_j'$ is
an index $\ell_j$ QC code.

For the converse, suppose $C'_j$ is an index $\ell_j$ QC code. So,
given $\beta_j \in V_j$, there exists $\gamma_j\in V_j$ such that
$$\left(\Tr_{\F_{q^n}/\F_q}\left((\beta_j\alpha^{-\ell_ji_{j}})\alpha^{ki_{j}}\right) \right)_{0\leq k \leq
N-1}=\left(\Tr_{\F_{q^n}/\F_q}\left(\gamma_j\alpha^{ki_{j}}\right)
\right)_{0\leq k \leq N-1},$$ or equivalently
$$\Tr_{\F_{q^n}/\F_q}\left((\beta_j\alpha^{-\ell_ji_{j}}-\gamma_j)\alpha^{ki_{j}}\right) = 0 \ \ \mbox{for all $0\leq k \leq N-1$.}
$$
By assumption $i_j$ has size $n$ $q$-cyclotomic coset mod $N$.
Therefore by Theorem \ref{cg}, the above equality holds for all $k$
if and only if $\beta_j\alpha^{-\ell_ji_{j}}-\gamma_j=0$. In other
words $V_j$ is closed under multiplication by
$\alpha^{-\ell_ji_{j}}$ and $\ell_j$ is the minimal such positive
integer. Hence the result follows.
\end{proof}

We need two more facts on finite fields. Proofs are clear, hence
omitted.

\begin{prop} \label{ff}
(i) Consider $\Fq \subseteq \F_{q^d} \subseteq \F_{q^n}$ and let
$\alpha$ be a primitive element in $\F_{q^n}$. Then
$$L_d=\dfrac{q^n-1}{q^d-1}$$ is the least positive integer such
that $\alpha^{L_d} \in \F_{q^d}$. Moreover, $\alpha^{L_d}$ is a
primitive element of $\F_{q^d}$.

(ii) For $\alpha^i \in \F_{q^n}$, the least positive integer
$\ell_d$ satisfying $\alpha^{i\ell_d} \in \F_{q^d}$ is
$$\ell_d=\frac{\lcm(i,L_d)}{i}.$$
\end{prop}

The following result describes the set of possible indices for QC
subcodes in $C$.
\begin{thm}\label{index theorem}
(i) Consider the $\Fq$-linear subcode $C'$ of the cyclic code $C$
(cf. (\ref{cyclic trace}) and (\ref{subcode})). For $1\leq j \leq
s$, let $\F_{q^{d_j}}$ be the largest intermediate field in
$\F_{q^n} / \Fq$ such that $V_j$ is an $\F_{q^{d_j}}$-subspace of
$\F_{q^n}$. Let $$L_j = \dfrac{q^n-1}{q^{d_j}-1} \ \mbox{\ and\ } \
\ell_j=\dfrac{\lcm(i_j,L_j)}{i_j}, \ \forall j$$ and let
$\ell=\lcm(\ell_1,\ldots,\ell_s)$. If $\ell \neq q^n-1 = N$, then
$C'$ is an index $\ell$ QC subcode of $C$.

(ii) Let $\F_{q^{d_1}},\ldots,\F_{q^{d_u}}$ be the intermediate
fields of the extension $\F_{q^n} / \Fq$, where $d_1=1$ and $d_u=n$.
Let $L_j = \dfrac{q^n-1}{q^{d_j}-1}$ for each $1\leq j \leq u$ and
define the integers
$$\begin{array}{ccc}
  \ell_1^1=\dfrac{\lcm(i_1,L_1)}{i_1} & \cdots & \ell_s^1=\dfrac{\lcm(i_s,L_1)}{i_s} \\
  \vdots &  & \vdots \\
  \ell_1^u=\dfrac{\lcm(i_1,L_u)}{i_1} & \cdots & \ell_s^u=\dfrac{\lcm(i_s,L_u)}{i_s}
\end{array}.$$
For $$\begin{array}{l}
I=\left\{\lcm\left(\ell_{t_1}^{a_1},\ldots,\ell_{t_m}^{a_m}\right);
1\leq m \leq s, \
1\leq a_1,\ldots,a_m \leq u, \ 1\leq t_1,\ldots,t_m \leq s, \right.\\
  \left. \hskip10cm  t_i \neq t_j \ \mbox{for any } i\neq j \right\}\setminus \{q^n-1\},\end{array}$$ $C$ has QC subcodes of
index $\ell$ for every $\ell \in I$ and has no QC subcode of
different index.
\end{thm}

\begin{proof}
(i) Note that if some $V_j=\{0\}$, then it can be considered
``maximally" as a vector space over $\F_{q^n}$. Then the index of
$C_j'$ should be 1 since $\{0\}$ is a vector space over any subfield
of $\F_{q^n}$, and in particular over $\F_q(\alpha^{i_j})$ (cf.
Proposition \ref{ff0}). In this case, indeed, $L_j=1$ and hence
$\ell_j=1$. Now suppose $V_j$ is not zero and let $\F_{q^{d_j}}$ be
the largest field over which it is a vector space. By Proposition
\ref{ff}, the least power of $\alpha^{i_j}$ that lies in
$\F_{q^{d_j}}$ is $\ell_j$. So, $\F_q(\alpha^{\ell_j i_j})\subseteq
\F_{q^{d_j}}$ and hence $V_j$ is an $\F_q(\alpha^{\ell_j
i_j})$-space. By maximality of $d_j$, any field over which $V_j$ has
a vector space structure must be contained in $\F_{q^{d_j}}$.
Therefore, $\ell_j$ is indeed the index of $C'_j$ by Proposition
\ref{ff0}. It is clear that the index of $C'$ is the least common
multiple of indices of all $C'_j$'s.

(ii) Part (i) shows that the index of a QC subcode of $C$ lies in
$I$. Take an element
$\ell=\lcm\left(\ell_{t_1}^{a_1},\ldots,\ell_{t_m}^{a_m}\right)$ of
$I$. Consider the subcode $C'$ where $V_{t_j}$ (for $j=1,\ldots ,m$)
is ``maximally" defined over the intermediate field
$\F_{q^{d_{a_j}}}$ (e.g. $V_{t_j}=\F_{q^{d_{a_j}}}$) and all other
$V_j$'s equal $\{0\}$. Then the index of $C'_{t_j}$ is
$\ell^{a_j}_{t_j}$ for all $j=1,\ldots ,m$. Hence the index of $C'$
is $\ell$. Therefore for any element of $I$, there is a QC subcode
of $C$ of that index.
\end{proof}

\begin{rem}
Let us note that the exponents $i_j$'s in the trace representation
of the cyclic code $C$ can be any representative of a $q$-cyclotomic
coset mod $N$. In other words, replacing $i_j$ by $qi_j$ in
(\ref{cyclic trace}) still yields the same code. Let us observe that
the choice of cyclotomic coset representatives in $C$'s trace
representation does not affect indices of QC subcodes either. For
this, let $j\in \{1,\ldots,s\}$ and note that $L_{a}$ is relatively
prime to $q$ for any $1\leq a \leq u$ (with the notation of Theorem
\ref{index theorem}). Then,
$$\ell_{qi_j}^{a}=\dfrac{\lcm(qi_j,L_{a})}{qi_j}=\dfrac{q \lcm(i_j,L_a)}{q i_j}=\dfrac{\lcm(i_j,L_a)}{i_j}=\ell_{i_j}^{a}.$$
\end{rem}

\begin{rem}
Let us now show that a different primitive element choice for
$\F_{q^n}$ does not change the set $I$ of possible indices for QC
subcodes. Let $\eta$ be another primitive element in $\F_{q^n}$.
Then $\eta= \alpha^r$ for some $1< r <q^n-1$ and $\gcd(r,q^n-1)=1$.
Then,
$$BZ(C^\bot)=\{\alpha^{i_1},\ldots,\alpha^{i_s}\}=\{\eta^{j_1},\ldots,\eta^{j_s}\},$$
for some $j_1,\ldots ,j_s$. Note that for each $t\in\{1,\ldots
,s\}$, this means $\alpha^{rj_t}=\alpha^{i_t}$. In other words,
$$rj_t=i_t+k_t(q^n-1), \ \ 1\leq t \leq s ,$$
for some integers $k_1,\ldots ,k_s$. For $a\in \{1,\ldots ,u\}$, let
$\nu$ be such that
$$\left(i_t+k_t(q^n-1)\right)\nu =\lcm\left(i_t+k_t(q^n-1),L_a \right).$$
Since $L_a$ divides $k_t(q^n-1)$, we conclude that $\nu$ is the
smallest positive integer so that $i_t\nu$ is a multiple of $L_a$
(i.e. $i_t\nu=\lcm(i_t,L_a)$). Hence,
$$\frac{\lcm\left(i_t+k_t(q^n-1),L_a \right)}{i_t+k_t(q^n-1)} = \frac{\left(i_t+k_t(q^n-1)\right)\nu}{i_t+k_t(q^n-1)}= \nu=\frac{\lcm(i_t,L_a)}{i_t}.$$
Hence, contributions of $i_t$ and $j_t$ to the index relative to the
intermediate field $\F_{q^{d_a}}$ are the same.
\end{rem}

\begin{rem}\label{example}
A natural question is whether a cyclic code has an index $\ell$ QC
subcode for every divisor $\ell$ of its length. This is not
necessarily the case as the following example shows. Consider the
binary cyclic code $C$ of length 15 with the trace representation
$$C=\left\{
\left(\Tr_{\F_{2^4}/\F_2}\left(\lambda\alpha^{k}\right)
\right)_{0\leq k \leq 14} ; \lambda \in \F_{2^4}\right\}. $$ With
the notation of Theorem \ref{index theorem}, we have $i_1=1$,
$d_1=1$, $d_2=2$, $d_3=4$. Moreover,
$$L_1=15, \ L_2=5, \ L_3=1,$$
and
$$\ell_1^1=\frac{\lcm(1,15)}{1}=15, \ \ell_1^2=\frac{\lcm(1,5)}{1}=5, \ \ell_1^3=\frac{\lcm(1,1)}{1}=1.$$
Hence, $C$ has only cyclic and index 5 QC subcodes but no index 3 QC
subcode, which is the other divisor of its length. An index 5 QC
subcode can be easily obtained from Theorem \ref{index theorem} and
here is an example:
$$
C'=\left\{ \left(\Tr_{\F_{2^4}/\F_2}\left(\beta\alpha^{k}\right)
\right)_{0\leq k \leq 14} ; \beta \in \F_{2^2}\right\}.
$$
Note that $C'$ is a rather ``small" code, let us list its nonzero
codewords:
\begin{eqnarray*}
c_1 &=& (0, 1, 1, 0, 1, 0, 1, 1, 1, 1, \textbf{0, 0, 0, 1, 0}),\\
c_2 &=& (\textbf{0, 0, 0, 1, 0}, 0, 1, 1, 0, 1, 0, 1, 1, 1, 1),\\
c_3 &=& (0, 1, 1, 1, 1, \textbf{0, 0, 0, 1, 0}, 0, 1, 1, 0, 1).
\end{eqnarray*}
Clearly, the codewords of $C'$ are invariant relative to 5-shift.

Observe that $C$ is the binary simplex code of length 15. We will
have a complete investigation of QC subcodes of the simplex code in
Section \ref{simplex}.
\end{rem}

\section{Enumeration of QC Subcodes}\label{enum}
We saw in Section \ref{indices} that QC subcodes of a cyclic code
$C$ and their indices are determined by the vector subspaces of
$\F_{q^n}$ that determine the coefficients of terms in the trace
representation and the maximal intermediate fields over which they
have vector space structure. Hence, enumeration of QC subcodes of a
fixed index in $C$ is clearly related to the number of vector
subspaces in $\F_{q^n}$. Due to trace, however, various choices of
coefficient vector spaces may yield the same subcode. The following
result shows that this is not the case in our setting.

We will continue considering the cyclic code $C$ with the trace
representation in (\ref{cyclic trace}). For vector spaces
$V_1,\ldots ,V_s \subseteq \F_{q^n}$, let us denote the subcode
determined by them as in (\ref{subcode}) by $C_V$.

\begin{thm} \label{enum-main}
Consider the cyclic code $C$ in (\ref{cyclic trace}) and let
$V=(V_1,\ldots ,V_s)$ and $W=(W_1,\ldots ,W_s)$ be collection of
vector subspaces in $\F_{q^n}$ that determine subcodes $C_V$ and
$C_W$ as in (\ref{subcode}). Then, $C_V=C_W$ if and only if
$V_j=W_j$ for all $j$. In particular, the number of QC subcodes of
index $\ell$ in $C$ is the same as the number of vector space
choices $V_1,\ldots ,V_s \subseteq \F_{q^n}$ that yield this index.
\end{thm}

\begin{proof}
Let $\lambda_j\in V_j$ for all $j$ and consider the codeword
$$c_\lambda=\left(\Tr_{\F_{q^n}/ \F_{q}}\left(\sum_{j=1}^{s} \lambda_j \alpha^{ki_j}\right)\right)_k \in C_V.$$
This codeword also belongs to $C_W$ if and only if there exists
$\beta_j\in W_j$ such that
$$c_\lambda=c_\beta=\left(\Tr_{\F_{q^n}/ \F_{q}}\left(\sum_{j=1}^{s} \beta_j \alpha^{ki_j}\right)\right)_k .$$
This holds exactly when
\begin{equation}\label{cond}
\Tr_{\F_{q^n}/ \F_{q}}\left((\lambda_1-\beta_1)x^{i_1}+\cdots +
(\lambda_s-\beta_s)x^{i_s}\right)=0 \ \mbox{for all $x\in
\F_{q^n}$}.
\end{equation}
Since each $i_j$ has full size $q$-cyclotomic coset by assumption,
Theorem \ref{cg} implies that (\ref{cond}) holds if and only if
$\lambda_j=\beta_j$ for all $j$.
\end{proof}

\begin{rem}
We can loosen the assumption on the cyclotomic cosets of $i_j$'s and
still write a criterion for equality of subcodes $C_V$ and $C_W$.
Namely, suppose that $$|Cyc_q(i_j)|=n \ \mbox{for $1\leq j \leq r$},
\ \ |Cyc_q(i_j)|=\delta_j <n \ \mbox{for $r+1\leq j \leq s$}.$$
Theorem \ref{cg} implies that $C_V=C_W$ if and only if
$$V_j=W_j \ \mbox{for $1\leq j \leq r$ and } \Tr_{\F_{q^n}/ \F_{q^{\delta_j}}}(V_j)=\Tr_{\F_{q^n}/ \F_{q^{\delta_j}}}(W_j) \ \mbox{for $r+1\leq j \leq s$ (cf. (\ref{cond}))}.$$
\end{rem}

Let us recall that the number of nonzero $\F_q$-subspaces in
$\F_{q^n}$ is determined by $q$-binomial coefficients as
\begin{equation} \label{binom}
N(n,q)= \dsum_{k=1}^{n}{n\choose k}_q =
\dsum_{k=1}^{n}\dfrac{(q^n-1)(q^n-q)\cdots(q^n-q^{k-1})}{(q^k-1)(q^k-q)\cdots(q^k-q^{k-1})}.
\end{equation}
Note that each $k\in \{1,\ldots ,n\}$ counts $\F_q$-subspaces of
dimension $k$. \textbf{In the rest of the manuscript, the number of
subspaces will refer to this number, which excludes the zero
subspace.}

The following result provides the number of nonzero vector subspaces
in $\F_{q^n}$ maximally defined over an intermediate field of the
extension $\F_{q^n}/\F_q$. It will be used in the following sections
to obtain enumeration results for QC subcodes of certain cyclic
codes. We first introduce some notation.

Let $n=u_1^{a_1}u_2^{a_2}\cdots u_t^{a_t}$, where $u_i$'s are
pairwise distinct prime numbers and $a_i$'s are nonnegative
integers. For $\displaystyle{\vec{i}=(i_1,\ldots ,i_t)\in
\prod_{1\leq j\leq t} \{0,1,\ldots , a_j\}}$, we denote the
intermediate field $\F_{q^{u_1^{i_1}\cdots u_t^{i_t}}}$ by
$\F_{i_1,\ldots ,i_t}$. For any $1\leq j_1<j_2<\cdots <j_v\leq t$,
we let
\begin{equation}\label{intsection}
N_{\vec{i}}(j_1,\ldots ,j_v):= N\left(\cdots
u_{i_{j_1}}^{a_{j_1}-i_{j_1}-1}\cdots
u_{i_{j_v}}^{a_{j_v}-i_{j_v}-1}\cdots ,q^{\cdots
u_{i_{j_1}}^{i_{j_1}+1}\cdots u_{i_{j_v}}^{i_{j_v}+1}\cdots}
\right).
\end{equation}
In (\ref{intsection}), only the terms corresponding to $j_\nu$'s are
written. For any $\mu \in \{1,\ldots ,t\}\setminus \{j_1,\ldots
,j_v\}$, the exponent of $u_\mu$ in the dimension part of the
expression is $a_\mu-i_\mu$ and the exponent of $u_\mu$ in the field
size part is $i_\mu$. Finally, we will assume that
\begin{equation}\label{conv}
N_{\vec{i}}(j_1,\ldots ,j_v)=0, \ \mbox{if $i_{j_\nu}+1>a_{j_\nu}$
for some $\nu \in \{1,\ldots ,v\}$}.
\end{equation}

\begin{thm}\label{max field}
Let $n=u_1^{a_1}u_2^{a_2}\cdots u_t^{a_t}$, where $u_i$'s are
pairwise distinct prime numbers and $a_i$'s are nonnegative
integers, and consider the extension $\F_{q^n}/\F_q$. For any
$\displaystyle{\vec{i}=(i_1,\ldots ,i_t)\in \prod_{1\leq j\leq t}
\{0,1,\ldots , a_j\}}$, the number of nonzero subspaces of
$\F_{q^n}$ that are maximally defined over $\F_{i_1,\ldots ,i_t}$ is
given by
\begin{equation}\label{incl-excl}
N(u_1^{a_1-i_1}\cdots u_t^{a_t-i_t},q^{u_1^{i_1}\cdots
u_t^{i_t}})-\sum_{1\leq j \leq t}N_{\vec{i}}(j)+\sum_{1\leq
j_1<j_2\leq t}N_{\vec{i}}(j_1,j_2)-\cdots
-(-1)^{t-1}N_{\vec{i}}(1,2,\ldots ,t).
\end{equation}
\end{thm}

\begin{proof}
The number of nonzero vector spaces in $\F_{q^n}$ defined over
$\F_{i_1,\ldots ,i_t}$ is $N(u_1^{a_1-i_1}\cdots
u_t^{a_t-i_t},q^{u_1^{i_1}\cdots u_t^{i_t}})$. Note that if a vector
space $V\subseteq \F_{q^n}$ is defined over an intermediate field
properly containing $\F_{i_1,\ldots ,i_t}$, then it also has a
vector space structure over $\F_{i_1,\ldots ,i_t}$. Therefore we
need to subtract the number of all such vector spaces from
$N(u_1^{a_1-i_1}\cdots u_t^{a_t-i_t},q^{u_1^{i_1}\cdots
u_t^{i_t}})$. An intermediate field properly containing
$\F_{i_1,\ldots ,i_t}$ has to contain at least one of the following
fields:
$$\F_{i_1+1,i_2, \ldots ,i_t}, \ \ \F_{i_1,i_2+1,i_3,\ldots ,i_t}, \ \ \F_{i_1,\ldots ,i_{t-1}, i_t+1}$$
For each $1\leq j \leq t$, set
$$S_j^{\vec{i}}:=\left\{V\subseteq \F_{q^n}: V\not=0, \ \mbox{$V$: vector space over $\F_{i_1,\ldots ,i_{j-1},i_{j}+1,i_{j+1},\ldots ,i_t}$} \right\}.$$
For any $1\leq j_1<j_2<\cdots <j_v\leq t$, note that
$S_{j_1}^{\vec{i}}\cap S_{j_2}^{\vec{i}}\cap \cdots \cap
S_{j_v}^{\vec{i}}$ consists of nonzero subspaces $V\subseteq
\F_{q^n}$ which are defined over the composite of the intermediate
fields
$$\F_{i_1,\ldots ,i_{j_1-1},i_{j_1}+1,i_{j_1+1},\ldots ,i_t}, \ \ldots \ ,\F_{i_1,\ldots ,i_{j_v-1},i_{j_v}+1,i_{j_v+1},\ldots ,i_t}.$$
Hence,
$$\left| S_{j_1}^{\vec{i}}\cap S_{j_2}^{\vec{i}}\cap \cdots \cap S_{j_v}^{\vec{i}} \right|=N_{\vec{i}}(j_1,\ldots ,j_v).$$
Observe that $S_{j_1}^{\vec{i}}\cap S_{j_2}^{\vec{i}}\cap \cdots
\cap S_{j_v}^{\vec{i}}=\emptyset$ if $i_{j_\nu}+1>a_{j_\nu}$ for
some $\nu \in \{1,\ldots ,v\}$ (cf. (\ref{conv})). Since the number
of nonzero subspaces in $\F_{q^n}$ that are defined maximally over
the subfield $\F_{i_1,\ldots ,i_t}$ is
$$N(u_1^{a_1-i_1}\cdots u_t^{a_t-i_t},q^{u_1^{i_1}\cdots u_t^{i_t}})-\left| S_1^{\vec{i}}\cup S_2^{\vec{i}}\cup \cdots \cup S_t^{\vec{i}} \right|,$$
and the inclusion-exclusion principle states that
$$\left| S_1^{\vec{i}}\cup S_2^{\vec{i}}\cup \cdots \cup S_t^{\vec{i}} \right|=\sum_{1\leq j \leq t}|S_j|-\sum_{1\leq j_1<j_2\leq t}|S_{j_1}\cap S_{j_2}|+\cdots +(-1)^{t-1}| S_1^{\vec{i}}\cap S_2^{\vec{i}}\cap \cdots \cap S_t^{\vec{i}}|,$$
the result follows.
\end{proof}

\section{QC Subcodes of the Simplex Code}\label{simplex}
The $q$-ary simplex code of length $N=q^n-1$ is defined as
$$C=\left\{\left(\Tr_{\F_{q^n}/\F_q}\left(\lambda\alpha^{k}\right)\right)_{0\leq k \leq N-1} ;
\lambda\in \F_{q^n} \right\}.$$ We will let $n$ be as general as
possible: $n=u_1^{a_1}u_2^{a_2}\cdots u_t^{a_t}$, where $u_i$'s are
pairwise distinct prime numbers and $a_i$'s are nonnegative
integers. For $\displaystyle{\vec{i}=(i_1,\ldots ,i_t)\in
\prod_{1\leq j\leq t} \{0,1,\ldots , a_j\}}$, let
\begin{equation}\label{simplex indices}
L_{\vec{i}}=L_{i_1,\ldots,i_t}:=\frac{q^n-1}{q^{u_1^{i_1}\cdots
u_t^{i_t}}-1}.
\end{equation}
By Theorem \ref{index theorem}, the set of indices of QC subcodes in
$C$ is
$$I=\left\{ L_{i_1,\ldots,i_t}: \ 0\leq i_j \leq a_j \ \mbox{for all $j$} \right\},$$
when $q\not=2$. For $q=2$, one has to exclude $L_{0,\ldots
,0}=2^n-1$ from the set above.

Let us denote the intermediate field $\F_{q^{u_1^{i_1}\cdots
u_t^{i_t}}}$ of $\F_{q^n}/\F_q$ by $\F_{i_1,\ldots ,i_t}$. By
Theorem \ref{enum-main}, the number of QC subcodes of index
$L_{i_1,\ldots,i_t}$ in the simplex code is equal to the number of
nonzero subspaces in $\F_{q^n}$ that are defined maximally over the
subfield $\F_{i_1,\ldots ,i_t}$. Using Theorem \ref{max field}, we
obtain the following.

\begin{thm}\label{simplex count}
Let $C$ be the $q$-ary simplex code of length $N=q^n-1$, where
$n=u_1^{a_1}u_2^{a_2}\cdots u_t^{a_t}$ for pairwise distinct primes
$u_i$ and nonnegative integers $a_i$. With the notations in
(\ref{simplex indices}), (\ref{intsection}) and (\ref{conv}), for
any $\displaystyle{\vec{i}=(i_1,\ldots ,i_t)\in \prod_{1\leq j\leq
t} \{0,1,\ldots , a_j\}}$, $C$ has
$$N(u_1^{a_1-i_1}\cdots u_t^{a_t-i_t},q^{u_1^{i_1}\cdots u_t^{i_t}})-\sum_{1\leq j \leq t}N_{\vec{i}}(j)+\sum_{1\leq j_1<j_2\leq t}N_{\vec{i}}(j_1,j_2)-\cdots -(-1)^{t-1}N_{\vec{i}}(1,2,\ldots ,t)$$
QC subcodes of index $L_{\vec{i}}$. For $q=2$, exclude
$\vec{i}=(0,\ldots ,0)$ from this conclusion.
\end{thm}

\section{QC Subcodes of Duals of BCH Codes}\label{bch}
\subsection{Dual of the binary double-error-correcting BCH code}
Dual of the binary double-error-correcting BCH code of length
$N=2^n-1$ is defined as
$$C=\left\{\left(\Tr_{\F_{2^n}/\F_2}\left(\lambda\alpha^{k} + \beta \alpha^{3k} \right)\right)_{0\leq k \leq N-1} ;
\lambda , \beta \in \F_{2^n} \right\}.$$ With the notation of
Theorem \ref{index theorem}, $i_1=1$ and $i_2=3$. Note that
$|Cyc_q(3)|=n$ for any $q$ and $n$ except for $q=2$ and $n=2$. We
will investigate this code for two families of $n$ values.

\subsubsection{\textbf{$n=u^{a-1}$: power of a prime}}\label{power of prime}
In order to have full cyclotomic coset for $3$, we will exclude the
case $u=2=a$, in which case $C$ is a rather short and uninteresting
code.

Intermediate fields of the extension $\F_{2^{u^{a-1}}}/\F_2$ are
\begin{equation}\label{intbch}
\F_2\subset \F_{2^u} \subset \F_{2^{u^2}} \subset \F_{2^{u^3}}
\subset \cdots \subset \F_{2^{u^{a-1}}}.
\end{equation}
We have
$$L_1=2^{u^{a-1}}-1=N, \ L_2=\frac{2^{u^{a-1}}-1}{2^u-1}, \ L_3=\frac{2^{u^{a-1}}-1}{2^{u^2}-1}, \ldots , L_a=1 \ (\mbox{cf. Theorem \ref{index theorem}}).$$
It is clear that $\ell_1^i=L_i$ for all $i=1,\ldots ,a$. Now we need
to compute $\ell_2^i$'s. Since $2^2\equiv 1$ (mod $3$), we have
$$
2^{u^{a-1}}-1 \equiv \left\{\begin{array}{ll} 1 \ \mbox{mod} \ 3, &
\mbox{if $u$ is odd,}\\ 0 \ \mbox{mod} \ 3, & \mbox{if $u=2$}.
\end{array}\right.
$$
Note that $L_i\mid L_1$ for all $i$. Hence for $u$ an odd prime, we
have $3\nmid L_i$ for all $i$ (since $3\nmid L_1$ in this case) and
$$\ell_2^i=\frac{\lcm(3,L_i)}{3}=L_i.$$
For $u=2$ however, $3\mid L_1$. Therefore $\ell_2^1=L_1/3=L_2$ in
this case. For $i>1$, we have
\begin{eqnarray*}
L_i&=&\frac{L_1}{2^{2^{i-1}}-1}\\
&=&1+2^{2^{i-1}}+2^{2\cdot 2^{i-1}}+2^{3\cdot 2^{i-1}}+\cdots +2^{(2^{a-i}-1)\cdot 2^{i-1}}\\
&\equiv&2^{a-i} \ \mbox{mod } 3 \\
&\not\equiv& 0 \ \mbox{mod } 3
\end{eqnarray*}
whether $a-i$ is odd or even. Hence, $3\nmid L_i$ for $i>1$ when
$u=2$. Therefore for $u=2$,
$$\ell_2^i=L_i \ \mbox{for $i=2,\ldots ,a$}.$$
Our conclusions are summarized in Table 1.

\begin{table}[h]
\begin{tabular}{|c|c|}
\multicolumn{2}{c}{$u>2$}     \\
  \hline
  $\ell_1^1=L_1$ & $\ell_2^1=L_1$ \\
  \hline
  $\ell_1^2=L_2$ & $\ell_2^2=L_2$ \\
  \hline
    $\ell_1^3=L_3$ & $\ell_2^3=L_3$ \\
  \hline
  $\vdots$ & $\vdots$ \\
  \hline
    $\ell_1^a=L_a$ & $\ell_2^a=L_a$ \\
  \hline
\end{tabular}
\quad \quad \quad \quad
\begin{tabular}{|c|c|}
\multicolumn{2}{c}{$u=2$}     \\
  \hline
  $\ell_1^1=L_1$ & $\ell_2^1=L_2$ \\
  \hline
  $\ell_1^2=L_2$ & $\ell_2^2=L_2$ \\
  \hline
    $\ell_1^3=L_3$ & $\ell_2^3=L_3$ \\
  \hline
  $\vdots$ & $\vdots$ \\
  \hline
    $\ell_1^a=L_a$ & $\ell_2^a=L_a$ \\
  \hline
\end{tabular}
\caption{$\ell_1^i$ and $\ell_2^i$ values for the dual of the binary
BCH code for $n=u^{a-1}$}
\end{table}
The following result describes all possible indices for QC subcodes
of the dual of the binary double-error-correcting BCH code when
$n=u^{a-1}$.

\begin{prop}\label{bchindices}
For the dual of the binary double-error-correcting BCH code $C$ of
length $N=2^{u^{a-1}}-1$, indices of QC subcodes are
$$I=\{L_a=1, L_{a-1},\ldots ,L_2\}.$$
\end{prop}

\begin{proof}
By Theorem \ref{index theorem}, $I$ consists of $\ell_1^i$'s,
$\ell_2^j$'s and other values that $\lcm(\ell_1^i,\ell_2^j)$'s can
produce. Of course, $L_1=N$ is excluded from $I$. By Table 1,
$L_2,L_3,\ldots ,L_a$ are all contained in $I$. We need to check the
outcomes of $\lcm(\ell_1^i,\ell_2^j)$'s.

It is well-known that the polynomial $x^{u^{t-1}}-1$ (for any $t$)
factors over $\mathbb{Q}$ into cyclotomic polynomials:
$$x^{u^{t-1}}-1=(x-1)\phi_u(x)\phi_{u^2}(x)\phi_{u^3}(x)\cdots \phi_{u^{t-1}}(x).$$
Therefore for any $1\leq i \leq a$, we have
$$L_i=\frac{2^{u^{a-1}}-1}{2^{u^{i-1}}-1}=\frac{\phi_u(2)\phi_{u^2}(2)\cdots \phi_{u^{a-1}}(2)}{\phi_u(2)\phi_{u^2}(2)\cdots \phi_{u^{i-1}}(2)}=\phi_{u^i}(2)\phi_{u^{i+1}}(2)\cdots \phi_{u^{a-1}}(2).$$
Hence for $i<j$, we have $L_j\mid L_i$ and $\lcm(L_i,L_j)=L_i$.
Using Table 1, $\lcm(\ell_1^i,\ell_2^j)=\lcm(L_i,L_j)$'s do not
bring any new value to $I$ for both odd and even prime values of
$u$.
\end{proof}

We are ready to count QC subcodes.

\begin{thm}\label{bch count 1}
Consider the dual of the binary double-error-correcting BCH  code
$C$ of length $N=2^{u^{a-1}}-1$ for a prime number $u$ and an
integer $a\geq 2$ except for $u=2=a$. Let
$$A_{j}=N(u^{a-j-1},2^{u^{j}})-N(u^{a-j-2},2^{u^{j+1}}) \ \ \mbox{for $0\leq j \leq a-1$.}$$

(i) If $u$ is odd, then $C$ has 3 cyclic subcodes (including
itself), $$3A_{a-2}+A_{a-2}N(u,2^{u^{a-2}})$$ QC subcodes of index
$L_{a-1}$, and
\begin{equation}\label{eq1}
2A_{j-1}+A_{j-1}(N(u^{a-j},2^{u^{j-1}})+N(u^{a-j-1},2^{u^{j}}))
\end{equation}
QC subcodes of index $L_j$ for any $2\leq j \leq a-2$.

(ii) If $u=2$, then $C$ has 3 cyclic subcodes (including itself) for
any $a$. For $a\geq 4$, there are
$$3A_{a-2}+A_{a-2}N(u,2^{u^{a-2}})$$ QC subcodes of index $L_{a-1}$,
\begin{equation}\label{eq2}
2A_{j-1}+A_{j-1}(N(u^{a-j},2^{u^{j-1}})+N(u^{a-j-1},2^{u^{j}}))
\end{equation}
QC subcodes of index $L_j$ for $3\leq j \leq a-2$, and
\begin{equation}\label{eq3}
2A_{1}+A_0+A_{1}N(u^{a-1},2)+(A_0+A_1)N(u^{a-3},2^{u^{2}})
\end{equation}
QC subcodes of index $L_2$. If $a=3$, there are
\begin{equation}\label{eq4}
3A_1+2A_0+A_1N(u^2,2)
\end{equation}
QC subcodes of index $L_2$.
\end{thm}

\begin{proof}
A QC subcode of $C$ is of the form
\begin{equation}\label{QCsubcode}C'=\left\{\left(\Tr_{\F_{2^{u^{a-1}}}/\F_2}\left(\lambda\alpha^{k} + \beta \alpha^{3k} \right)\right)_{0\leq k \leq N-1} ;
\lambda \in V , \beta \in W \right\},\end{equation} where
$V,W\subseteq \F_{2^{u^{a-1}}}$ are subspaces defined over some
intermediate field of the extension $\F_{2^{u^{a-1}}}/\F_2$. Note
that this extension's subfield structure is rather simple (cf.
(\ref{intbch})). Hence it is easy to see that $A_{j}$ describes the
number of nonzero subspaces in $\F_{2^{u^{a-1}}}$ defined maximally
over $\F_{2^{u^{j}}}$ for $0\leq j \leq a-1$ (cf. (\ref{incl-excl})
and note that $A_{a-1}=1$ by the convention in (\ref{conv})).

The following observation in the proof of Proposition
\ref{bchindices} will be the main tool for our analysis below:
$$\mbox{If $i<j$, then $\lcm(L_i,L_j)=L_i$.}$$
Note that in any case, subcodes of index $L_a=1$ (i.e. cyclic
subcodes) are obtained by letting $(V=\F_{2^{u^{a-1}}}, W=0),
(V=0,W=\F_{2^{u^{a-1}}})$ or
$(V=\F_{2^{u^{a-1}}},W=\F_{2^{u^{a-1}}})$ (the code $C$ itself).
Moreover when $a=2$, we have an extension $\F_{2^u}/\F_2$ and $C$
only has cyclic subcodes and there are 3 cyclic subcodes as noted
above. Hence, we will assume that $a\geq 3$ below.

(i) Choices of $V$ and $W$ that yield QC subcodes of index $L_{a-1}$
can be systematically listed as follows:
\[\begin{array}{l}
\mbox{$V$: maximally defined over $\F_{2^{u^{a-2}}}$}, \ W=0\\
V=0, \mbox{$W$: maximally defined over $\F_{2^{u^{a-2}}}$}\\
\mbox{$V$: maximally defined over $\F_{2^{u^{a-2}}}$}, \mbox{$W$: defined over $\F_{2^{u^{a-2}}}$}\\
\mbox{$V$: maximally defined over $\F_{2^{u^{a-1}}}$}, \mbox{$W$:
maximally defined over $\F_{2^{u^{a-2}}}$}
\end{array}\]
The number of such choices are, respectively,
$$A_{a-2}, \ A_{a-2}, \ A_{a-2}N(u,2^{u^{a-2}}), \ A_{a-1} A_{a-2},$$
whose sum is $3A_{a-2}+A_{a-2}N(u,2^{u^{a-2}})$. Note that for
$a=3$, there are only QC subcodes of index $L_{a-1}=L_2$ (other than
cyclic subcodes) whose count is given above and (\ref{eq1}) does not
apply to this situation. For $a\geq 4$ and $2\leq j \leq a-2$,
choices of $V$ and $W$ that yield QC subcodes of index $L_j$ are
\[\begin{array}{l}
\mbox{$V$: maximally defined over $\F_{2^{u^{j-1}}}$}, \ W=0\\
V=0, \mbox{$W$: maximally defined over $\F_{2^{u^{j-1}}}$}\\
\mbox{$V$: maximally defined over $\F_{2^{u^{j-1}}}$}, \mbox{$W$: defined over $\F_{2^{u^{j-1}}}$}\\
\mbox{$V$: maximally defined over $\F_{2^{u^{j}}}$}, \mbox{$W$: maximally defined over $\F_{2^{u^{j-1}}}$}\\
\mbox{$V$: maximally defined over $\F_{2^{u^{j+1}}}$}, \mbox{$W$: maximally defined over $\F_{2^{u^{j-1}}}$}\\
\hspace{5cm} \vdots \\
\mbox{$V$: maximally defined over $\F_{2^{u^{a-1}}}$}, \mbox{$W$: maximally defined over $\F_{2^{u^{j-1}}}$}\\
\end{array}\]
The number of choices for the first three combinations are,
respectively,
$$A_{j-1}, \ A_{j-1}, \ A_{j-1}N(u^{a-j},2^{u^{j-1}}).$$
The remaining choices add up to $N(u^{a-j-1},2^{u^{j}}) A_{j-1}$.
The total is
$$2A_{j-1}+A_{j-1}(N(u^{a-j},2^{u^{j-1}})+N(u^{a-j-1},2^{u^{j}})).$$

(ii) For $a\geq 5$, the count of $L_j$ for $3\leq j \leq a$ follows
as in part (i) and the results for index $L_{a-1}$ subcodes,
(\ref{eq2}) and (\ref{eq3}) are identical. For $a=4$, counting goes
similarly again but note that (\ref{eq2}) does not apply to this
situation. So, what needs special attention here is the number of
index $L_2$ QC subcodes for $a\geq 4$ and $a=3$. For $a\geq 4$,
choices of $V$ and $W$ that yield QC subcode of index $L_2$ are as
follows:
\[\begin{array}{l}
\mbox{$V$: maximally defined over $\F_{2^{u}}$}, \ W=0\\
V=0, \mbox{$W$: maximally defined over $\F_{2}$}\\
V=0, \mbox{$W$: maximally defined over $\F_{2^{u}}$}\\
\mbox{$V$: maximally defined over $\F_{2^{u}}$}, \mbox{$W$: defined over $\F_{2}$}\\
\mbox{$V$: maximally defined over $\F_{2^{u^{2}}}$}, \mbox{$W$: maximally defined over $\F_{2}$}\\
\mbox{$V$: maximally defined over $\F_{2^{u^{2}}}$}, \mbox{$W$: maximally defined over $\F_{2^u}$}\\
\mbox{$V$: maximally defined over $\F_{2^{u^{3}}}$}, \mbox{$W$: maximally defined over $\F_{2}$}\\
\mbox{$V$: maximally defined over $\F_{2^{u^{3}}}$}, \mbox{$W$: maximally defined over $\F_{2^u}$}\\
\hspace{5cm} \vdots \\
\mbox{$V$: maximally defined over $\F_{2^{u^{a-1}}}$}, \mbox{$W$: maximally defined over $\F_{2}$}\\
\mbox{$V$: maximally defined over $\F_{2^{u^{a-1}}}$}, \mbox{$W$: maximally defined over $\F_{2^u}$}\\
\end{array}\]
The number of choices for the first four combinations are,
respectively,
$$A_{1}, \ A_{0}, \ A_1, \ A_{1}N(u^{a-1},2).$$
The remaining choices add up to $(A_0+A_1)N(u^{a-3},2^{u^{2}})$.
Total is the desired value.

When $a=3$, choices of $V$ and $W$ that yield QC subcode of index
$L_2$ are as follows:
\[\begin{array}{l}
\mbox{$V$: maximally defined over $\F_{2^{u}}$}, \ W=0\\
V=0, \mbox{$W$: maximally defined over $\F_{2}$}\\
V=0, \mbox{$W$: maximally defined over $\F_{2^{u}}$}\\
\mbox{$V$: maximally defined over $\F_{2^{u}}$}, \mbox{$W$: defined over $\F_{2}$}\\
\mbox{$V$: maximally defined over $\F_{2^{u^{2}}}$}, \mbox{$W$: maximally defined over $\F_{2}$}\\
\mbox{$V$: maximally defined over $\F_{2^{u^{2}}}$}, \mbox{$W$: maximally defined over $\F_{2^u}$}\\
\end{array}\]
The number of choices for the first four combinations are,
respectively,
$$A_{1}, \ A_{0}, \ A_1, \ A_{1}N(u^2,2).$$
The remaining two choices add up to $(A_0+A_1)$. Hence the total is
$3A_1+2A_0+A_1N(u^2,2)$ and the proof is finished.
\end{proof}

\subsubsection{\textbf{$n=uv$: product of two distinct primes}} \label{distinct primes}
We consider the dual of the binary double-error-correcting BCH code
of length $N=2^{uv}-1$, where $u$ and $v$ are distinct primes.

We have
\[\begin{array}{cc}
L_1=2^{uv}-1 & L_2=\dfrac{2^{uv}-1}{2^u-1}=1+2^u+2^{2u}+\cdots +2^{(v-1)u} \\
L_3=\dfrac{2^{uv}-1}{2^v-1}=1+2^v+2^{2v}+\cdots +2^{(u-1)v} & L_4=1.
\end{array}\]
Depending on whether the product $uv$ is odd or even, we can compute
the values of $\ell_1^i$ and $\ell_2^j$ as in Section \ref{power of
prime}. The results are presented in Table 2.

\begin{table}[H]
\begin{tabular}{|c|c|}
\multicolumn{2}{c}{$u,v$: odd}     \\
  \hline
  $\ell_1^1=L_1$ & $\ell_2^1=L_1$ \\
  \hline
  $\ell_1^2=L_2$ & $\ell_2^2=L_2$ \\
  \hline
    $\ell_1^3=L_3$ & $\ell_2^3=L_3$ \\
  \hline
    $\ell_1^4=L_4$ & $\ell_2^4=L_4$ \\
  \hline
\end{tabular}
\quad \quad \quad \quad
\begin{tabular}{|c|c|}
\multicolumn{2}{c}{$u=2$, $v$: odd}     \\
  \hline
  $\ell_1^1=L_1$ & $\ell_2^1=L_1/3$ \\
  \hline
  $\ell_1^2=L_2$ & $\begin{array}{c} \ell_2^2=L_2/3 \ (v=3) \\ \ell_2^2=L_2 \ (v\not=3)\end{array}$ \\
  \hline
    $\ell_1^3=L_3$ & $\ell_2^3=L_3/3$ \\
  \hline
    $\ell_1^4=L_4$ & $\ell_2^4=L_4$ \\
  \hline
\end{tabular}
\caption{$\ell_1^i$ and $\ell_2^i$ values for the dual of the BCH
code for $n={uv}$}
\end{table}

The following result describes all possible indices for QC subcodes
in the case $n=uv$.

\begin{prop}\label{bchindices2}
For the dual of the binary double-error-correcting BCH code $C$ of
length $N=2^{uv}-1$, indices of QC subcodes are follows:
$$I=\left\{\begin{array}{ll}\{L_2,L_3,L_4\}& \mbox{if $u,v$ are both odd}\\ \{L_2,L_3,L_3/3,L_4\}& \mbox{if $u=2$, $v\not=3$: odd}\\ \{1,3,7,9,21\}& \mbox{if $u=2$, $v=3$}\\\end{array}\right.$$
\end{prop}

\begin{proof}
We use the $\ell_1^i$ and $\ell_2^i$ values in Table 2. Let us note
that
\begin{eqnarray*}
2^{uv}-1 &=& \phi_u(2)\phi_v(2)\phi_{uv}(2),\\
2^{u}-1 &=& \phi_u(2),\\
2^{v}-1 &=& \phi_v(2).
\end{eqnarray*}
Hence, $L_2=\phi_v(2)\phi_{uv}(2)$ and $L_3=\phi_u(2)\phi_{uv}(2)$.

For the case $u$ and $v$ are both odd, other than $L_2,L_3$ and
$L_4$, the only possible index for a QC subcode is $\lcm(L_2,L_3)$.
However,
$$\lcm(L_2,L_3)=\phi_{uv}(2)\lcm(2^u-1,2^v-1)=\phi_{uv}(2)(2^u-1)(2^v-1)=L_1.$$
Hence $I=\{L_2,L_3,L_4\}$ in this case.

For $u=2$, $v$ odd and $v\not=3$, note that $L_1/3=L_2$. Therefore,
other than $L_2,L_3,L_3/3$ and $L_4$, the only other possible index
for a QC subcode is $\lcm(L_2,L_3/3)$. Since $\phi_u(2)=2^u-1=3$ in
this case, we have $L_3/3=\phi_{uv}(2)$ and this is a divisor of
$L_2$. Hence $\lcm(L_2,L_3/3)=L_2$ and there is no new contribution
to $I$ in this case.

When $u=2$ and $v=3$, it is easy to verify the set $I$ is as stated
in the proposition.
\end{proof}

We are ready to count QC subcodes of the dual of the BCH code in the
case $n=uv$.

\begin{thm}\label{bch count 2}
Consider the dual of the binary double-error-correcting BCH  code
$C$ of length $N=2^{uv}-1$, where $u$ and $v$ are distinct prime
numbers. Then $C$ has 3 cyclic subcodes (including itself) for any
$u$ and $v$. Moreover:

(i) If $u$ and $v$ are odd, then $C$ has $(2N(v,2^u)+N(v,2^u)^2-3)$
QC subcodes of index $L_2$ and $(2N(u,2^v)+N(u,2^v)^2-3)$ QC
subcodes of index $L_3$.

(ii) If $u=2$ and $v\not=3$ and odd, then $C$ has
$(N(uv,2)+N(v,2^u)+N(uv,2)N(v,2^u)-2N(u,2^v)-1)$ QC subcodes of
index $L_2$, $(N(u,2^v)^2-1)$ QC subcodes of index $L_3$ and
$(2N(u,2^v)-2)$ QC subcodes of index $L_3/3$.

(iii) If $u=2$ and $v=3$, then $C$ has $124194$ QC subcodes of index
$21$, $99$ QC subcodes of index $9$, $84$ QC subcodes of index $7$
and $18$ QC subcodes of index $3$.
\end{thm}

\begin{proof}
Note that the number of subspaces in $\F_2^{uv}$ that are maximally
defined over $\F_2$, $\F_{2^u}$ and $\F_{2^v}$ are respectively
given by $(N(uv,2)-N(u,2^v)-N(v,2^u)+1)$, $(N(v,2^u)-1)$ and
$(N(u,2^v)-1)$. Consider a subcode $C'$ as in (\ref{QCsubcode}) of
$C$. We will use Table 2 and the combination of
$\lcm(\ell_1^i,\ell_2^j)$'s that lead to the corresponding index.
For all possibilities of $u$ and $v$, it is clear that there are 3
cyclic subcodes of $C$ obtained by the choices
$(V=\F_{2^{uv}},W=0)$, $(V=0,W=\F_{2^{uv}})$ and
$(V=\F_{2^{uv}},W=\F_{2^{uv}})$.

If $u$ and $v$ are both odd, then the choices of subspaces $V,W$
that yield the indices $L_2,L_3$ are as follows:
\[\begin{array}{ll}
L_2: & \mbox{$V$: maximally defined over $\F_{2^{u}}$}, \ W=0\\
&V=0, \mbox{$W$: maximally defined over $\F_{2^u}$}\\
&\mbox{$V$: maximally defined over $\F_{2^{u}}$}, \mbox{$W$: maximally defined over $\F_{2^{u}}$}\\
&\mbox{$V$: maximally defined over $\F_{2^{u}}$}, \mbox{$W$: maximally defined over $\F_{2^{uv}}$}\\
&\mbox{$V$: maximally defined over $\F_{2^{uv}}$}, \mbox{$W$:
maximally defined over $\F_{2^u}$}
\end{array}\]
Total number of such subspaces is
$$(N(v,2^u)-1)+(N(v,2^u)-1)+(N(v,2^u)-1)N(v,2^u)+(N(v,2^u)-1)=2N(v,2^u)+N(v,2^u)^2-3.$$
\[\begin{array}{ll}
L_3: & \mbox{$V$: maximally defined over $\F_{2^{v}}$}, \ W=0\\
&V=0, \mbox{$W$: maximally defined over $\F_{2^v}$}\\
&\mbox{$V$: maximally defined over $\F_{2^{v}}$}, \mbox{$W$: maximally defined over $\F_{2^{v}}$}\\
&\mbox{$V$: maximally defined over $\F_{2^{v}}$}, \mbox{$W$: maximally defined over $\F_{2^{uv}}$}\\
&\mbox{$V$: maximally defined over $\F_{2^{uv}}$}, \mbox{$W$:
maximally defined over $\F_{2^v}$}
\end{array}\]
Total number of such subspaces is $$2N(u,2^v)+N(u,2^v)^2-3.$$

If $u=2$ and $v\not=3$, the choices of subspaces $V,W$ that yield
the indices $L_2,L_3,L_3/3$ are alisted below. Note that $L_1/3=L_2$
when $u=2$. Moreover, $\lcm(L_2,L_3/3)=L_2$ (cf. proof of
Proposition \ref{bchindices}).
\[\begin{array}{ll}
L_2: & \mbox{$V$: maximally defined over $\F_{2^{u}}$}, \ W=0\\
&V=0, \mbox{$W$: maximally defined over $\F_{2}$}\\
&V=0, \mbox{$W$: maximally defined over $\F_{2^u}$}\\
&\mbox{$V$: maximally defined over $\F_{2^{u}}$}, \mbox{$W$: defined over $\F_{2}$}\\
&\mbox{$V$: maximally defined over $\F_{2^{uv}}$}, \mbox{$W$: maximally defined over $\F_{2}$}\\
&\mbox{$V$: maximally defined over $\F_{2^{uv}}$}, \mbox{$W$:
maximally defined over $\F_{2^u}$}
\end{array}\]
Total number of such subspaces is
$$
N(uv,2)+N(v,2^u)+N(uv,2)N(v,2^u)-2N(u,2^v)-1.$$
\[\begin{array}{ll}
L_3: & \mbox{$V$: maximally defined over $\F_{2^{v}}$}, \ W=0\\
&\mbox{$V$: maximally defined over $\F_{2^{v}}$}, \mbox{$W$: maximally defined over $\F_{2^v}$}\\
&\mbox{$V$: maximally defined over $\F_{2^{v}}$}, \mbox{$W$:
maximally defined over $\F_{2^{uv}}$}
\end{array}\]
Total number of such subspaces is $$N(u,2^v)^2-1.$$
\[\begin{array}{ll}
L_3/3: & V=0, \ \mbox{$W$: maximally defined over $\F_{2^v}$}\\
&\mbox{$V$: maximally defined over $\F_{2^{uv}}$}, \mbox{$W$:
maximally defined over $\F_{2^v}$}
\end{array}\]
Total number of such subspaces is $$2(N(u,2^v)-1).$$

Finally, if $u=2$ and $v=3$, we have
$$\ell_1^2=21,\ \ell_1^3=9,\ \ell_1^4=1, \ \ell_2^1=21,\ \ell_2^2=7,\ \ell_2^3=3,\ \ell_2^4=1.$$
Moreover, $N(3,4)=43$, $N(2,8)=10$ and $N(6,2)=2824$. Then,
\[\begin{array}{ll}
21: & \mbox{$V$: maximally defined over $\F_{2^{2}}$}, \ W=0\\
&V=0, \mbox{$W$: maximally defined over $\F_{2}$}\\
&\mbox{$V$: maximally defined over $\F_{2^{2}}$}, \mbox{$W$: defined over $\F_{2}$}\\
&\mbox{$V$: maximally defined over $\F_{2^{6}}$}, \mbox{$W$:
maximally defined over $\F_{2}$}
\end{array}\]
Total number of such subspaces is
\[\begin{array}{l}
\Bigl((N(3,4)-1)+\left(N(6,2)-N(3,4)-N(2,8)+1\right)+(N(3,4)-1)N(6,2)+ \Bigr.\\
\hspace{7cm} \Bigl. \left(N(6,2)-N(3,4)-N(2,8)+1\right)
\Bigr)=124194.\end{array}\]
\[\begin{array}{ll}
9: & \mbox{$V$: maximally defined over $\F_{2^{3}}$}, \ W=0\\
&\mbox{$V$: maximally defined over $\F_{2^{3}}$}, \mbox{$W$: maximally defined over $\F_{2^3}$}\\
&\mbox{$V$: maximally defined over $\F_{2^{3}}$}, \mbox{$W$:
maximally defined over $\F_{2^6}$}
\end{array}\]
Total number of such subspaces is
$$(N(2,8)-1)+(N(2,8)-1)N(2,8)=99.$$
\[\begin{array}{ll}
7: & V=0, \ \mbox{$W$: maximally defined over $\F_{2^2}$}\\
&\mbox{$V$: maximally defined over $\F_{2^{6}}$}, \mbox{$W$:
maximally defined over $\F_{2^2}$}
\end{array}\]
Total number of such subspaces is $$2(N(3,4)-1)=84.$$
\[\begin{array}{ll}
3: & V=0, \ \mbox{$W$: maximally defined over $\F_{2^3}$}\\
&\mbox{$V$: maximally defined over $\F_{2^{6}}$}, \mbox{$W$:
maximally defined over $\F_{2^3}$}
\end{array}\]
Total number of such subspaces is $$2(N(2,8)-1)=18.$$
\end{proof}

\subsection{Dual of the $p$-ary BCH Code of Designed Distance 3}\label{bch-p}
Let $p$ be an odd prime. Dual of the $p$-ary BCH code of length
$N=p^n-1$ and designed distance 3 has the following trace
representation:
$$C=\left\{\left(\Tr_{\F_{p^n}/\F_p}\left(\lambda\alpha^{k} + \beta \alpha^{2k} \right)\right)_{0\leq k \leq N-1} ;
\lambda , \beta \in \F_{p^n} \right\}.$$ \textbf{As in Section
\ref{distinct primes}, we will consider the case $n=uv$ where $u$
and $v$ are distinct prime numbers.} In this case  $|Cyc_p(2)|=n$
for any odd prime $p$ and hence Theorem \ref{enum-main} applies. We
have
\begin{eqnarray*}
L_1&=&\frac{p^{uv}-1}{p-1}=\phi_u(p)\phi_v(p)\phi_{uv}(p)\\
L_2&=&\frac{p^{uv}-1}{p^u-1}=1+p^u+p^{2u}+\cdots +p^{(v-1)u}=\phi_v(p)\phi_{uv}(p) \\
L_3&=&\frac{p^{uv}-1}{p^v-1}=1+p^v+p^{2v}+\cdots +p^{(u-1)v}=\phi_u(p)\phi_{uv}(p) \\
L_4&=&1.
\end{eqnarray*}
Note that any positive power of $p$ is congruent to 1 mod 2.
Therefore $L_2\equiv v$ and $L_3\equiv u$ mod 2. This implies that
when $u$ and $v$ are both odd primes, no $L_i$ is divisible by 2. If
$u=2$ and $v$ is an odd prime, then $L_1$ and $L_3$ are divisible by
2, $L_2$ is not. Combining these observations, values of $\ell_1^i$
and $\ell_2^j$ are presented in Table 3. Note that when $u=2$,
$L_1/2$ and $L_3/2$ are not equal to $L_i$ for any $1\leq i \leq 4$
regardless of choice of the odd primes $p$ and $v$.

\begin{table}[h]
\begin{tabular}{|c|c|}
\multicolumn{2}{c}{$u,v$: odd}     \\
  \hline
  $\ell_1^1=L_1$ & $\ell_2^1=L_1$ \\
  \hline
  $\ell_1^2=L_2$ & $\ell_2^2=L_2$ \\
  \hline
    $\ell_1^3=L_3$ & $\ell_2^3=L_3$ \\
  \hline
    $\ell_1^4=L_4$ & $\ell_2^4=L_4$ \\
  \hline
\end{tabular}
\quad \quad \quad \quad
\begin{tabular}{|c|c|}
\multicolumn{2}{c}{$u=2$, $v$: odd}     \\
  \hline
  $\ell_1^1=L_1$ & $\ell_2^1=L_1/2$ \\
  \hline
  $\ell_1^2=L_2$ & $\ell_2^2=L_2$ \\
  \hline
    $\ell_1^3=L_3$ & $\ell_2^3=L_3/2$ \\
  \hline
    $\ell_1^4=L_4$ & $\ell_2^4=L_4$ \\
  \hline
\end{tabular}
\caption{$\ell_1^i$ and $\ell_2^i$ values for the dual of the
$p$-ary BCH code for $n={uv}$}
\end{table}

The following result describes all possible indices for QC subcodes
in this case.

\begin{prop}\label{bchindices3}
For the dual of the $p$-ary BCH code $C$ of length $N=p^{uv}-1$ and
designed distance 3, indices of QC subcodes are follows:
$$I=\left\{\begin{array}{ll}\{L_1,L_2,L_3,L_4\}& \mbox{if $u,v$ are both odd}\\ \{L_1,L_1/2,L_2,L_3,L_3/2,L_4\}& \mbox{if $u=2$, $v$: odd}
\end{array}\right.$$
\end{prop}

\begin{proof}
By Table 3, all index values in the statement belong to $I$. Again,
we need to check that $\lcm(\ell_1^i,\ell_2^j)$ values do not bring
any different index to $I$. One of the key observations for this
purpose is the following:
$$\lcm(\phi_u(p),\phi_v(p))=\phi_u(p)\phi_v(p)=\frac{(p^u-1)(p^v-1)}{(p-1)^2}.$$
Therefore $\lcm(L_2,L_3)=L_1$ in any case. Moreover, it is clear
that $L_i\mid L_1$ for $i=2,3,4$. Therefore the result follows
immediately for the case $u$ and $v$ are both odd.

For $u=2$, $v$: odd case, the following extra least common multiple
values, compared to previous case, can be easily verified:
\[\begin{array}{ccc}\label{lcms}
\lcm(L_1,L_1/2)=L_1 & \lcm(L_1,L_3/2)=L_1 & \lcm(L_2,L_1/2)=L_1/2 \\
\lcm(L_2,L_3/2)=L_1/2 & \lcm(L_3,L_1/2)=L_1 & \lcm(L_3,L_3/2)=L_3
\end{array}\]
Therefore the result also follows in the second case.
\end{proof}

We count the QC subcodes in the following result.

\begin{thm}\label{bch count 3}
Consider the dual of the $p$-ary BCH  code $C$ of length
$N=p^{uv}-1$, where $u$ and $v$ are distinct prime numbers. Let
$A=N(uv,p)-N(u,p^v)-N(v,p^u)+1$. Then $C$ has 3 cyclic subcodes
(including itself) for any $u$ and $v$. Moreover:

(i) If $u$ and $v$ are odd, then $C$ has
$$A+A\bigl(N(uv,p)+N(u,p^v)+N(v,p^u)\bigr)+2\bigl(N(u,p^v)N(v,p^u)-N(u,p^v)-N(v,p^u)+1 \bigr)$$
QC subcodes of index $L_1$,
$$2N(v,p^u)+N(v,p^u)^2-3$$ QC subcodes of index $L_2$, and
$$2N(u,p^v)+N(u,p^v)^2-3$$ QC subcodes of index $L_3$.

(ii) If $u=2$ and $v$ an odd prime, $C$ has
$$A+AN(uv,p)+(N(u,p^v)-1)(A+N(v,p^u)-1)$$ QC subcodes of index $L_1$,
$$2A+(N(v,p^u)-1)(A+N(u,p^v)-1)$$ QC subcodes of index $L_1/2$,
$$2N(v,p^u)+N(v,p^u)^2-3$$ QC subcodes of index $L_2$,
$$N(u,p^v)^2-1$$ QC subcodes of index $L_3$ and
$$2(N(u,p^v)-1)$$ QC subcodes of index $L_3/2$.
\end{thm}

\begin{proof}
Let us note that the number of subspaces in $\F_p^{uv}$ that are
maximally defined over $\F_p$, $\F_{p^u}$ and $\F_{p^v}$ are given
by $A$, $(N(v,p^u)-1)$ and $(N(u,p^v)-1)$, respectively. Consider a
subcode $C'$ of $C$.
$$C'=\left\{\left(\Tr_{\F_{p^n}/\F_p}\left(\lambda\alpha^{k} + \beta \alpha^{2k} \right)\right)_{0\leq k \leq N-1} ;
\lambda \in V, \beta \in W \right\}.$$ We will proceed as in the
proof of Theorem \ref{bch count 2} and use Table 3 and the
combination of $\lcm(\ell_1^i,\ell_2^j)$'s that lead to the
corresponding index (cf. proof of Proposition \ref{bchindices3}).
For all possibilities of $u$ and $v$, it is clear that there are 3
cyclic subcodes of $C$.

If $u$ and $v$ are both odd, then the choices of subspaces $V,W$
that yield index $L_1$ are as follows:
\[\begin{array}{ll}
L_1: & \mbox{$V$: maximally defined over $\F_{p}$}, \ W=0\\
& V=0 \ \mbox{$W$: maximally defined over $\F_{p}$}\\
&\mbox{$V$: maximally defined over $\F_{p}$}, \mbox{$W$: defined over $\F_{p}$}\\
&\mbox{$V$: maximally defined over $\F_{p^{u}}$}, \mbox{$W$: maximally defined over $\F_{p}$}\\
&\mbox{$V$: maximally defined over $\F_{p^{u}}$}, \mbox{$W$: maximally defined over $\F_{p^v}$}\\
&\mbox{$V$: maximally defined over $\F_{p^{v}}$}, \mbox{$W$: maximally defined over $\F_{p}$}\\
&\mbox{$V$: maximally defined over $\F_{p^{v}}$}, \mbox{$W$: maximally defined over $\F_{p^u}$}\\
&\mbox{$V$: maximally defined over $\F_{p^{uv}}$}, \mbox{$W$:
maximally defined over $\F_{p}$}
\end{array}\]
Total number of such subspaces is
$$2A+AN(uv,p)+(N(v,p^u)-1)(A+N(u,p^v)-1)+(N(u,p^v)-1)(A+N(v,p^u)-1)+A,$$
which yields the desired result. For index $L_2$, we have
\[\begin{array}{ll}
L_2: & \mbox{$V$: maximally defined over $\F_{p^u}$}, \ W=0\\
& V=0 \ \mbox{$W$: maximally defined over $\F_{p^u}$}\\
&\mbox{$V$: maximally defined over $\F_{p^u}$}, \mbox{$W$: maximally defined over $\F_{p^u}$}\\
&\mbox{$V$: maximally defined over $\F_{p^{u}}$}, \mbox{$W$: maximally defined over $\F_{p^{uv}}$}\\
&\mbox{$V$: maximally defined over $\F_{p^{uv}}$}, \mbox{$W$: maximally defined over $\F_{p^u}$}\\
\end{array}\]
Total number of such subspaces is
$$2(N(v,p^u)-1)+(N(v,p^u)-1)N(v,p^u)+(N(v,p^u)-1)=2N(v,p^u)+N(v,p^u)^2-3.$$
The result for index $L_3$ follows identically.

When $u=2$ and $v$ an odd prime, the choices of $V,W$ yielding index
$L_1$ are as follows:
\[\begin{array}{ll}
L_1: & \mbox{$V$: maximally defined over $\F_{p}$}, \ W=0\\
&\mbox{$V$: maximally defined over $\F_{p}$}, \mbox{$W$: defined over $\F_{p}$}\\
&\mbox{$V$: maximally defined over $\F_{p^{v}}$}, \mbox{$W$: maximally defined over $\F_{p}$}\\
&\mbox{$V$: maximally defined over $\F_{p^{v}}$}, \mbox{$W$: maximally defined over $\F_{p^u}$}\\
\end{array}\]
Total number of such subspaces is
$$A+AN(uv,p)+(N(u,p^v)-1)(A+N(v,p^u)-1).$$ For the other indices, we
have the following:
\[\begin{array}{ll}
L_1/2: & V=0, \ \mbox{$W$: maximally defined over $\F_{p}$}\\
&\mbox{$V$: maximally defined over $\F_{p^u}$}, \mbox{$W$: maximally defined over $\F_{p}$}\\
&\mbox{$V$: maximally defined over $\F_{p^u}$}, \mbox{$W$: maximally defined over $\F_{p^v}$}\\
%&\mbox{$V$: maximally defined over $\F_{p^{v}}$}, \mbox{$W$: maximally defined over $\F_{p}$}\\
%&\mbox{$V$: maximally defined over $\F_{p^{v}}$}, \mbox{$W$: maximally defined over $\F_{p^u}$}\\
&\mbox{$V$: maximally defined over $\F_{p^{uv}}$}, \mbox{$W$: maximally defined over $\F_{p}$}\\
\end{array}\]
Total number of such subspaces is
$$A+(N(v,p^u)-1)(A+N(u,p^v)-1)+A.$$
\[\begin{array}{ll}
L_2: & \mbox{$V$: maximally defined over $\F_{p^u}$}, \ W=0\\
&V=0, \ \mbox{$W$: maximally defined over $\F_{p^u}$}\\
&\mbox{$V$: maximally defined over $\F_{p^u}$}, \mbox{$W$: maximally defined over $\F_{p^u}$}\\
&\mbox{$V$: maximally defined over $\F_{p^u}$}, \mbox{$W$: maximally defined over $\F_{p^{uv}}$}\\
&\mbox{$V$: maximally defined over $\F_{p^{uv}}$}, \mbox{$W$: maximally defined over $\F_{p^u}$}\\
\end{array}\]
Total number of such subspaces is
$$2(N(v,p^u)-1)+(N(v,p^u)-1)N(v,p^u)+(N(v,p^u)-1).$$
\[\begin{array}{ll}
L_3: & \mbox{$V$: maximally defined over $\F_{p^v}$}, \ W=0\\
&\mbox{$V$: maximally defined over $\F_{p^v}$}, \mbox{$W$: maximally defined over $\F_{p^v}$}\\
&\mbox{$V$: maximally defined over $\F_{p^v}$}, \mbox{$W$: maximally defined over $\F_{p^{uv}}$}\\
\end{array}\]
Total number of such subspaces is $$2(N(u,p^v)-1)+(N(u,p^v)-1)^2.$$
\[\begin{array}{ll}
L_3/2: & V=0, \ \mbox{$W$: maximally defined over $\F_{p^v}$}\\
&\mbox{$V$: maximally defined over $\F_{p^{uv}}$}, \mbox{$W$: maximally defined over $\F_{p^v}$}\\
\end{array}\]
Total number of such subspaces is $2(N(u,p^v)-1).$
\end{proof}

\section{Examples}\label{examples}
Our results in Sections \ref{indices} and \ref{enum} yield an
algorithm, which can compute possible indices of QC subcodes of
a given cyclic code based on Theorem \ref{index theorem}, together
with the number of these subcodes by using Theorems \ref{enum-main}
and \ref{max field}. In \cite{Program}, we provide a
Magma code (\cite{BCP}), which illustrates the algorithm for length $q^n-1$ cyclic codes of the form
$$C=\left\{
\left(\Tr_{\F_{q^n}/\F_q}\left(\lambda_1\alpha^{ki_{1}}+\lambda_2\alpha^{ki_{2}}+
\lambda_3\alpha^{ki_{3}}\right) \right)_{0\leq k \leq q^n-2} ;
\lambda_j \in \F_{q^n},\ 1 \leq j \leq 3\right\},$$ where
$$n=u_1^{a_1}u_2^{a_2}u_3^{a_3}, \mbox{\ such\ that\ } a_i\geq 1,
u_i: \mbox{distinct\ primes\ for\ } 1 \leq i \leq 3.$$

In the following examples, we used the code in \cite{Program} for
some binary ($q=2$) and ternary ($q=3$) cyclic codes satisfying the
conditions above. Namely, Table \ref{table:binary table} presents
the indices and the appearances of QC subcodes of binary simplex
codes ($i_1=1, i_2=i_3=0$), dual of the double-error-correcting BCH
codes ($i_1=1, i_2=3, i_3=0$) and dual of the
triple-error-correcting BCH codes ($i_1=1, i_2=3, i_3=5$), for
various $n$ values. The results are listed in the form $[i,M_i]$, where $M_i$ is the number of proper nonzero QC subcodes of index $i$. In particular, for $i=1$ the corresponding count $M_1$ is the number of proper nonzero cyclic subcodes. Table \ref{table:ternary table} presents the
similar results for ternary simplex codes ($i_1=1, i_2=i_3=0$), dual
of the BCH codes of designed distance 3 ($i_1=1, i_2=2, i_3=0$) and
dual of the BCH codes of designed distance 5 ($i_1=1, i_2=2,
i_3=4$). Note that we do not consider prime $n$ values, since no
proper QC subcode occurs in this case. We do not consider cases where $q$-cyclotomic
coset mod $q^n-1$ for some $i_j$ has size less than $n$.

\begin{table}[H]
\begin{center}\small
\begin{tabular}{|C{0.5cm}|C{2.8cm}|C{5.1cm}|C{6.6cm}|}
\hline $n\backslash C$ & Simplex & Dual of Double-E.-C. BCH & Dual of Triple-E.-C. BCH \\\hline
6 & [1,0], [9,9], [21,42] & [1,2], [3,18], [7,84], [9,99],
[21,124194] & [1,6], [3,36], [7,168], [9,1287], [21,5468988]
\\\hline 8 & [1,0], [17,17], [85,510] & [1,2], [17,357],
[85,220697910] & [1,6], [17,190961], [51,150417870],
[85,116749194390]\\\hline 9 & [1,0], [73,146] & [1,2], [73,21900] &
[1,6], [73,3241784] \\\hline 10 & [1,0], [33,33], [341,12276] &
[1,2], [11,66], [33,1155], [341,2820939318120] & [1,6], [11,132],
[33,42735], [341,34635492948736680] \\\hline 12 & [1,0], [65,65],
[273,546], [585,5910], [1365,565721] & [1,2], [65,4485], [91,1092],
[195,391950], [273,299208], [455,37897860], [585,34614450],
[1365,276172787737667730] & [1,6], [13,260], [65,300495],
[91,73164], [117,23400], [195,26260650], [273,169888806360],
[455,2539156620], [585,207132845400], [819,146601246105077400],
[1365,156237298018977998951310]\\\hline 14 &  [1,0], [129,129],
[5461,51409854] & [1,2], [43,258], [129,16899],
[5461,209432100625503796112058] & [1,6], [43,516], [129,2247567],
[5461,10766874134934660085587731025396] \\\hline 15 & [1,0],
[1057,2114], [4681,617892] & [1,2], [1057,4477452],
[4681,381792995232] & [1,6], [1057,9474296888],
[4681,235907600998352976]\\\hline 16 & [1,0], [257,257],
[4369,78642], [21845,9370980720] & [1,2], [257,67077],
[4369,6225300720], $[21845,2\cdot 3\cdot 5\cdot 17\cdot 257\cdot
9632900474097094857135899]$ & [1,6], [257,17373971],
[4369,58338292825807289442], [13107,838758021781294495526312038290],
$[21845,2\cdot 3\cdot 5\cdot 7\cdot 11\cdot 17 \cdot 59 \cdot 257
\cdot 2062747\cdot 9632900474097094857135899]$ \\\hline 18 & [1,0],
[513,513], [4161,8322], [37449,195118125], $[87381, 2^3\cdot 3\cdot
5^2\cdot 11\cdot 19\cdot 73\cdot 370091]$ & [1,2], [171,1026],
[513,264195], [1387,16644], [4161,69272328], [12483,1624093999146],
[29127,28200771586760472], [37449,38069459320434786], $[87381,2\cdot
3^2\cdot 5^3\cdot 7\cdot 19\cdot 23\cdot 73\cdot 911\cdot
106077265549\cdot 1237940881586443]$ & [1,6], [171,2052],
[513,136588815], [1387,33288], [4161,576761402928],
[12483,13518958457429676], [29127,234743222688194168928],
[37449,7428358488736862865257616], $[87381, 2^2\cdot 3^2\cdot 5\cdot
7\cdot 19\cdot 73\cdot 2382323\cdot 2528261\cdot 25131697\cdot
143372569\cdot 5369043671723807]$ \\\hline 20 & [1,0], [1025,1025],
[33825,1151070], [69905,36070980], $[349525,5\cdot 7\cdot 41\cdot
43\cdot 39921132101]$ & [1,2], [1025,1054725], [11275,1181101350],
[33825,1323796103850], [69905,1301115742444320], $[349525, 2\cdot
3^2\cdot 5^2\cdot 11\cdot 17\cdot 19^2\cdot 31\cdot 41\cdot
5113\cdot 4182209\cdot 188408588933\cdot 147641569892759]$& [1,6],
[205,4100], [1025,1083202575], [6765,4600200],
[11275,1212991086450], [13981,144283920],
[33825,1525150786425400200],
[69905,3205081938871577654256028295040], $[209715, 2^3\cdot 3\cdot
5^2\cdot 11\cdot 31\cdot 37\cdot 41\cdot 43\cdot
908930777956680878604236175349273961]$, $[349525,2\cdot 3^2\cdot
5^2\cdot 11\cdot 31\cdot 41\cdot 89\cdot 126963961\cdot
\scriptstyle{795792305106258988205007754270563107304398999}]$\\\hline
\end{tabular}
\caption{Duals of binary BCH codes with designed distances 1, 3 and
5} \label{table:binary table}
\end{center}
\end{table}

\begin{table}[H]
\begin{center}\small
\begin{tabular}{|C{0.5cm}|C{2.9cm}|C{5.3cm}|C{6.3cm}|}
\hline $n\backslash C$ & Simplex & Dual of BCH with $d=3$ & Dual of
BCH with $d=5$
\\\hline 4 & [1,0], [10,10], [40,200] & [1,2], [5,20], [10,120],
[20,2400], [40,42400] & [1,6], [5,280], [10,30240], [20,508800],
[40,8988800]
\\\hline 6 & [1,0], [28,28], [91,182], [364,56630] & [1,2], [14,56],
[28,840], [91,33852], [182,10386376], [364,3196762296] & [1,6],
[7,112], [14,1680], [28,25200], [91,1917332872], [182,588204415344],
[364,181039089892752]\\\hline 8 & [1,0], [82,82], [820,9020],
[3280,127893760] & [1,2], [41,164], [82,6888], [410,757680],
[820,82118080], [1640,1164344791040], [3280,16357978191728640] &
[1,6], [41,14104], [82,578592], [205,1515360], [410,6960048480],
[820,10600942580628480], [1640,148923033457497538560],
[3280,2092232259971634166824960] \\\hline 9 & [1,0], [757,1514],
[9841,13721227572] & [1,2], [757,2298252],
[9841,188272127685375013488] & [1,6], [757,3484156088],
[9841,2583324994856249282153532653376] \\\hline 10 & [1,0],
[244,244], [7381,1225246], $[29524,2\cdot 3^2\cdot 17\cdot 61\cdot
136334867]$ &  [1,2], [122,488], [244,60024], [7381,1501232661500],
[14762,3118042234365339160], $[29524,2^3\cdot 5^2\cdot 11^2\cdot
61\cdot 105542903\cdot 41566356211]$ & [1,6], [61,976],
[122,120048], [244,14765904], [7381,3820376850953979715484712],
$[14762,2^4\cdot 3^2\cdot 7\cdot 11^2\cdot 19\cdot 61\cdot 71\cdot
191\cdot 4139226000747340297]$, $[29524,2^4\cdot 11^2\cdot 19\cdot
61\cdot 10103\cdot 16361\cdot 239527\cdot 40363307\cdot 4596044119]$
\\\hline 12 & [1,0], [730,730], [6643,13286], [20440,593480],
[66430,540023486], $[265720, 2\cdot 3^3\cdot 8378452950363007]$ &
[1,2], [365,1460],  [730,534360], [6643,176570940],
[10220,433900320], [20440,351798317920], [33215,7175655536140],
[66430,291618191759043240], $[132860,2^5\cdot 3\cdot 5\cdot 7\cdot
13\cdot 73\cdot 76623988461520162739]$, $[265720,2^5\cdot 5\cdot
7\cdot 13\cdot 73\cdot 192588767759642498898919144667]$ & [1,6],
[365,1071640], [730,391151520], [5110,317615034240],
[6643,2346274703864], [10220,257516368717440],
[20440,208789487298976640], [33215,3875117882212049705960],
$[66430,2^5\cdot 3\cdot 5\cdot 7\cdot 11\cdot 13\cdot 19\cdot
73\cdot 190845833\cdot 918851623\cdot 1129023677]$,
$[132860,2^7\cdot 3^3\cdot 5^3\cdot 7\cdot 13\cdot 19\cdot 73\cdot
2027338364818403272678960130077589]$, $[265720,2^7\cdot 5^2\cdot
7\cdot 13\cdot 31\cdot 73\cdot 181757219444968838257\cdot
773223758237056637579813]$ \\\hline 15 & [1,0], [59293,118586],
[551881,806850022], $[7174453,2\cdot 179\cdot 4561\cdot 357509\cdot
3559979471071921]$ & [1,2], [59293,14063113740],
[551881,651006961228800572], $[7174453,2^2\cdot 5\cdot 11^2\cdot
13\cdot 367\cdot 4561\cdot 101209\cdot 822407\cdot 100842919\cdot
9770548580137061374107091]$ & [1,6], [59293,1667716532673464],
[551881,525264982291624814236813816], $[7174453,2^3\cdot 11^2\cdot
13\cdot 521\cdot 4561\cdot 9993125731\cdot 152373840083\cdot
4006805689324561\cdot 13019832459914677\cdot 3778337670974685409]$
\\\hline
\end{tabular}
\caption{Duals of ternary BCH codes with designed distances 1, 3 and
5}\label{table:ternary table}
\end{center}
\end{table}


\begin{thebibliography}{20}

\bibitem {Program} ``Magma code for quasi-cyclic subcodes", \url{http://people.sabanciuniv.edu/~guneri/QCsubcodes.html}, February 22, 2016.

\bibitem {BCP} W. Bosma, J. Cannon and C. Playoust, ``The Magma algebra system. I. The user language", \emph{J. Symbolic Comput.}, vol. 24, 235-265, 1997.

\bibitem{C} Chen, E.Z., ``New quasi-cyclic codes from simplex codes", \emph{IEEE Trans. Inform. Theory}, vol. 53, 1193-1196,  2007.

\bibitem{DH} Daskalov, R.N., Hristov, P., ``New binary one-generator quasi-cyclic codes", \emph{IEEE Trans. Inform. Theory}, vol. 49, 3001-3005,  2003.

\bibitem{D} Dey, B.K., ``On the existence of good self-dual quasi-cyclic codes", \emph{IEEE Trans. Inform. Theory}, vol. 50, 1794-1798,  2004.

\bibitem {G} G\"{u}neri, C., ``Artin-Schreier curves and weights of two-dimensional cyclic codes", \emph{Finite Fields Appl.}, vol. 10, 481-505, 2004.

\bibitem{GM} G\"{u}neri, C., McGuire, G., ``Supersingular curves over finite fields and weight divisibility of codes", \emph{J. Comput. Appl. Math.}, vol. 259, part B, 474-484, 2014.

\bibitem{K} Kasami, T., ``A Gilbert-Varshamov bound for quasi-cyclic codes of rate 1/2", \emph{IEEE Trans. Inform. Theory}, vol. 20, p. 679,  1974.

\bibitem{LS1} Ling, S., Sol\'{e}, P., ``Good self-dual quasi-cyclic codes exist", \emph{IEEE Trans. Inform. Theory}, vol. 49, 1052-1053,  2003.

\bibitem{MW} Mart\'{i}nez-P\'{e}rez, C., Willems, W., ``Self-dual doubly even 2-quasi-cyclic transitive codes are asymptotically good", \emph{IEEE Trans. Inform. Theory}, vol. 53, 4302-4308,  2007.

\bibitem {W} Wolfmann, J., ``New bounds on cyclic codes from algebraic curves", in: Lecture Notes in Computer Science, vol. 388, 47-62, New York: Springer-Verlag, 1989.


\end{thebibliography}
\end{document}